\documentclass[12pt]{article}
\usepackage{fancyhdr}
\usepackage{amsmath}
\usepackage[title]{appendix}
\numberwithin{equation}{section}
\usepackage{verbatim}
\usepackage{amssymb}
\usepackage{mathtools}
\usepackage{amsthm}
\usepackage{color}
\usepackage{url}
\usepackage{graphicx}
\graphicspath{{images/} }
\usepackage[font=small,labelfont=bf]{caption}
\usepackage{natbib}
\usepackage{float}
\usepackage[font=scriptsize]{caption}
\usepackage{moresize}
\usepackage[font=scriptsize]{caption}
\RequirePackage[colorlinks,citecolor=blue,urlcolor=blue]{hyperref}
\usepackage{xurl}

\newtheorem{theorem}{Theorem}[section]

\newtheorem{definition}{Definition}[section]

\newcommand{\goto}{\rightarrow}

\newcommand{\V}{\mathbf{V}}
\newcommand{\transpose}{^{\tiny \top}}
\newcommand{\one}{\mathbf{1}}
\newcommand{\tree}{\mathbb{T}}
\newcommand{\TV}{\text{TV}}
\newcommand{\var}{\mathrm{Var}}
\newcommand{\cov}{\mathrm{Cov}}
\newcommand{\Ef}{{\mathrm{I}\!\mathrm{E}}}
\newcommand{\mb}{\mathbf}
\newcommand{\mbb}{\mathbb}
\newcommand{\mc}{\mathcal}

\begin{document}
	
\title{\textbf{When can we reconstruct the ancestral state? Beyond Brownian motion}}

\author{\small\sc Nhat L. Vu$^a$\footnote{These authors contributed equally to this work.}, Thanh P. Nguyen$^{b,c,d}$,\textsuperscript{*} Binh T. Nguyen$^{b,c,d}$, Vu Dinh$^{e}$, Lam Si Tung Ho$^{a}$}
\date{}
\maketitle
\thispagestyle{empty}

\begin{center}
\textit{\small a-Department of Mathematics and Statistics, Dalhousie University, Halifax, Nova Scotia, Canada\\
b-AISIA Research Lab, Ho Chi Minh City, Vietnam\\
c-Department of Computer Science, University of Science, Ho Chi Minh City, Vietnam\\
d-Vietnam National University, Ho Chi Minh City, Vietnam\\
e-Department of Mathematical Sciences, University of Delaware, USA}
\end{center}

\hrule
\begin{abstract}
Reconstructing the ancestral state of a group of species helps answer many important questions in evolutionary biology.
Therefore, it is crucial to understand when we can estimate the ancestral state accurately.
Previous works provide a necessary and sufficient condition, called the big bang condition, for the existence of an accurate reconstruction method under discrete trait evolution models and the Brownian motion model.
In this paper, we extend this result to a wide range of continuous trait evolution models. 
In particular, we consider a general setting where continuous traits evolve along the tree according to stochastic processes that satisfy some regularity conditions. 
We verify these conditions for popular continuous trait evolution models including Ornstein-Uhlenbeck, reflected Brownian Motion, bounded Brownian Motion, and Cox-Ingersoll-Ross.

	\medskip
	\textbf{Keywords:} ancestral state reconstruction, consistency, big bang condition, Ornstein-Uhlenbeck, Brownian motion, Cox-Ingersoll-Ross.
\end{abstract}

\section{Introduction}

The main task of the ancestral state reconstruction problem is to estimate the trait value of the most recent common ancestor of a group of species from their observed trait values.
This problem has many applications in ecology and evolution, such as reconstructing ancestral diets \citep{maritz2021repeated}, studying female song in songbirds \citep{odom2014female}, and inferring the origin of infectious disease epidemics \citep{faria2014early, gill2017relaxed}.
Moreover, ancestral state reconstruction methods have been applied to reconstruct the place of origin of a language family \citep{bouckaert2012mapping, neureiter2021can}.
Therefore, it is important to assess whether the ancestral state can be reconstructed with high certainty.
In particular, we are interested in whether reconstruction methods converge to the true ancestral trait value (in probability) as the number of sampled species increases to infinity.
This property is called \emph{consistency}, a desired characteristic for any reconstruction method.

The consistency property of ancestral state reconstruction methods has been studied previously for both discrete and continuous traits.
Since species are related to each other according to an evolutionary tree, it is natural that consistency depends heavily on the structure of this tree.
For continuous traits, \citet{ane2008analysis} derived a necessary and sufficient condition for the consistency of the Maximum likelihood estimator (MLE) of the ancestral state under the Brownian motion (BM) model.
This condition is $\one\transpose \V_n^{-1} \one \to \infty$ where $n$ is the number of sampled species, $\one$ is an $n \times 1$ vector whose elements are $1$, and $\V_n$ is an $n \times n$ matrix whose $i$th-row and $j$-th column is the distance from the root to the most recent common ancestor of leaves $i$ and $j$.
Under mild assumptions,  \citet{fan2018necessary} provided a necessary and sufficient condition, called big bang, for the existence of a consistent ancestral state reconstruction method for discrete traits.
The big bang condition is satisfied if for any positive number $\epsilon$, the number of intersection points between the circle with radius $\epsilon$ centred at the root and the tree converges to infinity as the number of sampled species increases.
In other words, the number of branches in the proximity of the root is large.
A natural question is whether the condition $\one\transpose \V_n^{-1} \one \to \infty$ and the big bang condition are equivalent.
Recently, \citet{ho2021can} gave a positive answer to this question when the sequence of trees is nested and has bounded height as the sample size increases.
In particular, they showed that both conditions are necessary and sufficient conditions for the existence of a consistent ancestral state reconstruction method under the BM model.
The result unifies the theory of ancestral state reconstruction between discrete trait evolution models and the BM model.
What remains unknown is whether this unified theory also holds for other evolution models of continuous traits.
In this paper, we will address this open question.

Researchers often model the evolution of a continuous trait along an evolutionary tree using a diffusion process.
At a node of the tree, the children lineages inherit the trait value of the parent lineage as their starting trait value and evolve independently from each other.
The simplest evolution model for continuous traits is the BM model, which only considers neutral evolution \citep{felsenstein1985phylogenies}.
Later, \citet{hansen1997stabilizing} proposed the Ornstein-Uhlenbeck (OU) model to incorporate natural selection. The BM and OU models are the most popular evolution models for continuous traits due to the existence of efficient computational methods \citep{ho2014linear}.
Recently, much effort has been made to move away from these Gaussian models \citep{boucher2016inferring, boucher2018general, blomberg2020beyond, jhwueng2020modeling}.
Following this spirit, we consider a general setting where the evolution of traits follows a general stochastic process on trees. 
We will show that under mild assumptions, the big bang condition is a necessary and sufficient condition for the existence of a consistent ancestral state reconstruction method.
Our setting includes several popular evolution models for continuous traits such as the OU model \citep{hansen1997stabilizing}, the reflected Brownian motion (RBM) model \citep{boucher2016inferring}, the bounded Brownian motion (BMM) model  \citep{boucher2016inferring}, and the Cox-Ingersoll-Ross (CIR) model \citep{lepage2006continuous, blomberg2020beyond}.
%Although the bounded Brownian motion (BBM) model \citep{boucher2016inferring} does not fall into our setting, we hypothesize that the big bang condition is also a necessary and sufficient condition for the existence of a consistent ancestral state reconstruction method under this model.
%We use simulation to provide strong supporting evidence for this hypothesis.

\section{Mathematical formulation}

In this paper, we consider the common setting for studying the asymptotic theory of trait evolution models where the sequence of trees $(\tree_n)_{n=1}^\infty$ is nested.
That is, $\tree_n$ is a subtree of $\tree_{n+1}$ for all $n$.
Without loss of generality, we assume that tree $\tree_n$ has $n$ species and all trees in the sequence have the same root.
We make a standard assumption that the tree topology and edge lengths of $\tree_n$ are known.
In this paper, we consider the scenario where the height of the sequence of trees $(\tree_n)_{n=1}^\infty$ is uniformly bounded from above. 
Specifically, let $t_i$ be the distance from the root to a leaf $i$.
The height of tree $\tree_n$ is defined by $h_n = \max\{t_i: i = 1, 2, \ldots, n\}$.		
Under our setting,  $h^* := \sup_n h_n < +\infty.$

For a tree $\tree$, we denote the leaf set of $\tree$ by $\partial \tree$ and the tree obtained by truncating $\tree$ at distance $s$ from the root by $\tree(s)$.
It is worth noticing that $\partial \tree(s)$ is called a cutset corresponding to time $s$ away from the root \citep{ane2017phase}.
Let $\lvert A \rvert$ be the number of elements of the set $A$.
The big bang condition \citep{fan2018necessary}  is defined as follows:

\begin{definition}[big bang condition]
A sequence of nested trees $(\tree_n)_{n=1}^\infty$ satisfies the big bang condition if $\lim_{n \to \infty} \lvert \partial \tree_n(s) \rvert = \infty$ for all $s > 0.$
\end{definition}

A layman's explanation of the big bang condition is that the number of branches in the proximity of the root is large.
We will prove that the big bang condition is the necessary and sufficient condition for the existence of a consistent ancestral state reconstruction method for a general class of continuous trait evolution models.
Throughout this paper, we assume that all other parameters of the model are known.
Specifically, we assume that the model satisfies the following regularity conditions:

	\begin{itemize}
		\item [\hypertarget{A1}{($A_1$)}]
		The trait evolves along a tree according to a time-homogeneous Markovian stochastic process $\{Y(t)\}_{t \geq 0}$ on a state-space $\mathcal{S} \subset \mbb{R}$ and there exist functions $u, v, \phi$ such that
\begin{equation}\label{mean}
		\Ef \{\phi[Y(t)] \mid Y(0)\} = u(t)\phi[Y(0)] + v(t),
		\end{equation}
		and
		\begin{equation}	\label{var}
		\var \{\phi[Y(t)] \mid Y(0) \} \leq C t, \quad \forall t \in [0,h^*],
		\end{equation}
where C is a constant that does not depend on $t$; $u, v$ are continuous functions; $\phi$ is a continuous, injective function from $\mc{S}$ to $\mbb{R}$ and the inverse function $\phi^{-1}$ is also continuous; and $u(t) > 0$ for all $t \ge 0$.
Note that $u(0) = 1$ and $v(0) = 0$.
		
		\item[\hypertarget{A2}{($A_2$)}]
		We denote the trait values at the leaves of a tree $\tree$ by $\mb{Y}_{\tree}$ and the trait value at the root by $\rho$.
		Let $P_{\rho, \tree}$ be the joint distribution of the observations $\mb Y_{\tree}$ at the leaves of the tree $\tree$ given that the ancestral state is $\rho$. 
		We assume that for any tree $\tree$, there exists $\rho_1 \ne \rho_2$ such that the overlapped support of $P_{\rho_1, \tree}$ and $P_{\rho_2, \tree}$ is not trivial.
		Here, an overlapped support $\mc{A}$ of $P_{\rho_1, \tree}$ and $P_{\rho_2, \tree}$ is trivial if $\min \{P_{\rho_1, \tree}(\mc{A}), P_{\rho_2, \tree}(\mc{A})\} = 0$.
		%That is, $\min \{P_{\rho_1, \tree}(\mc{A}), P_{\rho_2, \tree}(\mc{A})\} > 0$ where $\mc{A}$ is the overlapped support.
		Furthermore, we assume that $P_{\rho, \tree}$ admits a probability density function $f_{\rho, \tree}$.
		
	\end{itemize}	
	
Condition \hypertarget{A1}{($A_1$)} ensures that the observations contain sufficient information about the root. Specifically, \eqref{mean} means that the information about the root is contained in the expected value of the observations. On the other hand, \eqref{var} controls the decay rate of the information about the root through time. Condition \hypertarget{A2}{($A_2$)} is similar to the Downstream Disjointness condition in \cite{fan2018necessary} . The main purpose of this condition is to remove trivial scenarios. If \hypertarget{A2}{($A_2$)} does not hold, $P_{\rho_1, \tree}$ and $P_{\rho_2, \tree}$ are disjoint for any $\rho_1$, $\rho_2$. In this case, it is trivial to reconstruct the ancestral state because each observed vector of values at the leaves only corresponds to one value at the root. However, that is too good to be true for practical settings. 

\section{A necessary and sufficient condition for the existence of a consistent estimator for the ancestral state}

We recall the definition of a \emph{consistent estimator}:
\begin{definition} 
An estimator $\hat{\rho}_n$ of $\rho$ is consistent if and only if for all $\epsilon > 0$, we have 
\[
P_{\rho, \tree_n}(
\lvert \hat{\rho}_n - \rho \rvert > \epsilon) \goto 0.
\]
In other words, $\hat{\rho}_n$ converges to $\rho$ in probability ($\hat{\rho}_n \to_p \rho$).
\end{definition}

Now, we are ready to state our main result:

\begin{theorem}
Assume that the regularity condition \hyperlink{A1}{$(A_1)$} and \hyperlink{A2}{$(A_2)$} are satisfied.
Then, the big bang condition is the necessary and sufficient condition for the existence of a consistent ancestral reconstruction method.
\label{thm:main}
\end{theorem}

First, let us review briefly about properties of the total variation distance.
For any two probability measures $\mu_1$ and $\mu_2$, the total variation distance between them is defined as
\begin{align*}
d_\TV(\mu_1, \mu_2) &=\sup_{\mathcal{A}} \lvert \mu_1(\mathcal{A})-\mu_2(\mathcal{A}) \rvert\\
&=\frac{1}{2}\int \lvert \mu_1(x)-\mu_2(x) \rvert dx\\
&=\frac{1}{2}\int[\mu_1(x)\vee\mu_2(x)-\mu_1(x)\wedge\mu_2(x)]dx,
\end{align*}
where $a\vee b:=\max\{a,b\}$ and $a\wedge b:=\min\{a,b\}$. 
On the other hand, since
$$\frac{1}{2}\int[\mu_1(x)\vee\mu_2(x)+\mu_1(x)\wedge\mu_2(x)]dx = 1,$$
we have
\begin{equation}\label{eq1}
d_\TV(\mu_1, \mu_2) = 1 - \int\mu_1(x)\wedge\mu_2(x)dx.
\end{equation}
This implies that if the overlapped support of $\mu_1$ and $\mu_2$ is not trivial, then their total variation distance is strictly less than one.

The proof of Theorem \ref{thm:main} is divided into two parts.
In Section \ref{sub1}, we prove that the big bang condition is a necessary condition.
The main idea is to show that if the big bang condition does not hold, the total variation distance between $P_{\rho_1, \tree_n}$ and $P_{\rho_2, \tree_n}$ is bounded away from $1$ for some $\rho_1 \ne \rho_2$. 
This implies that there is no consistent estimator for the ancestral state $\rho$.
In Section \ref{sub2}, we provide a simple consistent estimator for the ancestral state when the big bang condition is satisfied.

\subsection{Necessary condition}
\label{sub1}
\begin{theorem}
Given that the condition \hyperlink{A2}{$(A_2)$} holds, then the big bang is the necessary condition for the existence of a consistent estimator for the ancestral state.
\label{them1} 
\end{theorem}
\begin{proof}
	We will use the contra-positive approach: no big bang condition implies no consistent estimator for the ancestral state. 
	That is, we need to prove that if there exists $s_0 > 0$, $n_0\ge1$, and $K>0$ such that  $ \lvert \partial\mathbb{T}_n(s_0) \rvert =K$ for all $n\ge n_0$, then there is no consistent estimator for the ancestral state.
	
	\medskip
	Recall that $P_{\rho, \tree_n}$ is the joint distribution of the observations $\mb Y_{\tree_n}$ at the leaves of tree $\tree_n$ with the ancestral state $\rho$.
	Note that, since $ \lvert \partial\mathbb{T}_n(s_0) \rvert = K$ for all $n\ge n_0$, $\mathbb{T}_n(s_0)$ is a fixed $K$-species star tree with edge lengths all equal to $s_0$ for all $n \ge n_0$. 
	From \hyperlink{A2}{$(A_2)$} the overlapped support of $P_{\rho_1,\mathbb{T}_n(s_0)}$ and $P_{\rho_2, \mathbb{T}_n(s_0)}$ is not trivial for some ancestral state $\rho_1\ne\rho_2$.
	Since this overlapped support is fixed for all $n \geq n_0$, by \eqref{eq1}, we have $d_\TV(P_{\rho_1,\mathbb{T}_n(s_0)}, P_{\rho_2,\mathbb{T}_n(s_0)}) \leq C_0 < 1$ for all $n \geq n_0$.
	
	\medskip
	Let $f_{\mathbb{T} \mid \mathbb{T}(s)}(\cdot \mid \tau)$ be the conditional probability density function of $\mb Y_{\tree}$ given $\mb Y_{\tree(s)} = \tau$.
	Note that  $f_{\mathbb{T} \mid \mathbb{T}(s)}(\cdot \mid \tau)$ does not depend on the ancestral state $\rho$ at the root of $\tree$ due to Markov property.
	We have
	\begin{align}
	&d_\TV(P_{\rho_1,\mathbb{T}_n}, P_{\rho_2,\mathbb{T}_n}) 
	=\frac{1}{2}\int \lvert \ f_{\rho_1, \mathbb{T}_n}(y) - f_{\rho_2, \mathbb{T}_n}(y) \rvert dy \nonumber\\
	&= \frac{1}{2}\int\left \lvert \int \left [ f_{\mathbb{T}_n \mid \mathbb{T}_n(s_0)}(y \mid \tau)f_{\rho_1, \mathbb{T}_n(s_0)}(\tau)  - f_{\mathbb{T}_n \mid \mathbb{T}_n(s_0)}(y \mid \tau)f_{\rho_2, \mathbb{T}_n(s_0)}(\tau) \right ] d\tau \right \rvert dy\nonumber\\
	&=\frac{1}{2}\int\left \lvert \int  f_{\mathbb{T}_n \mid \mathbb{T}_n(s_0)}(y \mid \tau) \left[f_{\rho_1, \mathbb{T}_n(s_0)}(\tau) - f_{\rho_2, \mathbb{T}_n(s_0)}(\tau)\right]d\tau\right \rvert dy\nonumber\\
	&\le \frac{1}{2}\int\int f_{\mathbb{T}_n \mid \mathbb{T}_n(s_0)}(y \mid \tau) \left \lvert f_{\rho_1, \mathbb{T}_n(s_0)}(\tau) - f_{\rho_2, \mathbb{T}_n(s_0)}(\tau) \right \rvert d \tau dy\nonumber\\
	&=\frac{1}{2}\int\left \lvert f_{\rho_1, \mathbb{T}_n(s_0)}(\tau) - f_{\rho_2, \mathbb{T}_n(s_0)}(\tau) \right \rvert \int f_{\mathbb{T}_n \mid \mathbb{T}_n(s_0)}(y \mid \tau) dy d\tau\nonumber \\
	&=\frac{1}{2}\int\left \lvert f_{\rho_1, \mathbb{T}_n(s_0)}(\tau) - f_{\rho_2, \mathbb{T}_n(s_0)}(\tau) \right \rvert d\tau \nonumber\\
	&= d_\TV(P_{\rho_1,\mathbb{T}_n(s_0)}, P_{\rho_2,\mathbb{T}_n(s_0)}) \leq C_0 < 1.
	\label{eq2}
	\end{align}
	Now suppose that there exists a consistent estimator $\hat{\rho}_n$ for the ancestral state $\rho$. Consider the event $\mathcal{A}_\epsilon:=\{\lvert \hat{\rho}_n -\rho_1 \rvert \leq \epsilon\}$ where $\epsilon$ is sufficiently small. 
	Then, we have
	$$P_{\rho_1,\mathbb{T}_n}(\mathcal{A}_\epsilon)\goto1\quad\text{and }\quad P_{\rho_2,\mathbb{T}_n}(\mathcal{A}_\epsilon)\goto 0,$$
	for any $\rho_1\ne \rho_2$. 
	Thus,
	\begin{align*}
	d_\TV(P_{\rho_1,\mathbb{T}_n}, P_{\rho_2,\mathbb{T}_n})=\sup_{\mathcal{B}}\left \lvert P_{\rho_1,\mathbb{T}_n}(\mathcal{B}) - P_{\rho_2,\mathbb{T}_n}(\mathcal{B})\right \rvert
	\ge \left \lvert P_{\rho_1,\mathbb{T}_n}(\mathcal{A_\epsilon}) - P_{\rho_2,\mathbb{T}_n}(\mathcal{A}_\epsilon)\right \rvert
	\goto 1,
	\end{align*}
	which is a contradiction to (\ref{eq2}). Therefore, there is no consistent estimator for the ancestral state in the absence of the big bang condition.\\
\end{proof}

\subsection{Sufficient condition}
\label{sub2}
\begin{theorem}\label{them2}
	Assume that the condition \hyperlink{A1}{$(A_1)$} is satisfied. 
	Then, the big bang condition implies that there exists a consistent estimator for the ancestral state.
\end{theorem}

\begin{proof}
We will construct a consistent estimator for the ancestral state $\rho$.
Under big bang condition, there exists an increasing sequence $\{n_k\}_{k=1}^\infty$ such that $ \lvert \partial\mathbb{T}_{n_k}(1/k) \rvert \ge k$ for all $k$.
This implies that there exists a $k$-species subtree of $\mathbb{T}_{n_k}$ such that the distances from the root to all internal nodes are at most $1/k$.
Let $(Y_{i, k})_{i=1}^k$ be the observations at the leaves and $(t_{i,k})_{i=1}^k$ be the distance from the root to the $i$-th leaf of this subtree.
Denote the trait value of the most recent common ancestor of $i$-th, $j$-th leaves of the subtree by $Y_{ij,k}$, and the distance from the root to this ancestor by $t_{ij,k}$.
It is worth noticing that $t_{i, k} \leq h^*$ and $t_{ij,k} \leq 1/k$.
We consider the following estimator for the ancestral state $\rho$:
\[
\hat{\rho} = \phi^{-1} \left (\frac{1}{k} \sum_{i=1}^k{\frac{\phi(Y_{i,k}) - v(t_{i,k})}{u(t_{i,k})}} \right ).
\]
Since $\phi$ is a continuous, injection function, the proposed estimator is well-defined. 
Moreover, $\phi^{-1}$ is a continuous function.
So, to prove that $\hat{\rho}$ is consistent, it is sufficient to show
\[
\frac{1}{k} \sum_{i=1}^k{\frac{\phi(Y_{i,k}) - v(t_{i,k})}{u(t_{i,k})}} \to_p \phi(\rho).
\]
Indeed, we have
\begin{align*}
\Ef \left (\frac{1}{k} \sum_{i=1}^k{\frac{\phi(Y_{i,k}) - v(t_{i,k})}{u(t_{i,k})}} \right ) &= \frac{1}{k} \sum_{i=1}^k{\frac{\Ef[\phi(Y_{i,k})] - v(t_{i,k})}{u(t_{i,k})}} \\
& = \frac{1}{k} \sum_{i=1}^k{\frac{[u(t_{i,k})\phi(\rho) + v(t_{i,k})] - v(t_{i,k})}{u(t_{i,k})}} \\
&= \phi(\rho).
\end{align*}
On the other hand,
\[
\var \left (\frac{1}{k} \sum_{i=1}^k{\frac{\phi(Y_{i,k}) - v(t_{i,k})}{u(t_{i,k})}} \right )  \leq \frac{1}{k^2 m^2}\left ( \sum_{i=1}^k{\var[\phi(Y_{i,k})]} + \sum_{1 \leq i < j \leq k}{\cov[\phi(Y_{i,k}), \phi(Y_{j, k})]} \right )
\]
where $m := \min_{t \in[0,h^*]} u(t) > 0$.
Note that $\var[\phi(Y_{i,k})] \leq C t_{i,k} \leq C h^*$ and
\begin{align*}
\cov[\phi(Y_{i,k}), \phi(Y_{j, k})] &= \cov\{\Ef[\phi(Y_{i,k}) \mid Y_{ij,k}], \Ef[\phi(Y_{j, k}) \mid Y_{ij,k}]\} \\
& \quad + \Ef \{ \cov[\phi(Y_{i,k}), \phi(Y_{j, k}) \mid Y_{ij, k}]\} \\
&= u(t_{ij,k})^2 \var[\phi(Y_{ij,k})]] \leq M^2 C t_{ij,k} \leq \frac{C M^2}{k},
\end{align*}
where $M := \max_{t \in[0,h^*]} u(t) < \infty$.
Therefore
\[
\var \left (\frac{1}{k} \sum_{i=1}^k{\frac{\phi(Y_{i,k}) - v(t_{i,k})}{u(t_{i,k})}} \right ) \leq \frac{C h^*}{k m^2} + \frac{CM^2}{k m^2} \to 0 \quad \text{as } k \to \infty.
\]
For any $\epsilon > 0$, by applying Chebyshev's inequality, we have
\[
\Pr \left ( \left \lvert \frac{1}{k} \sum_{i=1}^k{\frac{\phi(Y_{i,k}) - v(t_{i,k})}{u(t_{i,k})}} - \phi(\rho) \right \rvert > \epsilon \right ) \leq \frac{1}{\epsilon^2} \var \left (\frac{1}{k} \sum_{i=1}^k{\frac{\phi(Y_{i,k}) - v(t_{i,k})}{u(t_{i,k})}} \right ) \to 0.
\]
Hence,
\[
\frac{1}{k} \sum_{i=1}^k{\frac{\phi(Y_{i,k}) - v(t_{i,k})}{u(t_{i,k})}} \to_p \phi(\rho),
\]
which means $\hat \rho$ is consistent.

\end{proof}

We note that the estimator in this proof is sub-optimal since it only uses a subset of the observations.

\section{Applications}
In this section, we will apply our results to several popular trait evolution models including Ornstein-Uhlenbeck (OU), reflected Brownian motion (RBM) model, bounded Brownian motion (BBM) model, and Cox-Ingersoll-Ross (CIR) models.
Since we focus on the ancestral state reconstruction problem, we consider the classical setting where other parameters of the models are known.

\subsection{Ornstein-Uhlenbeck (OU) model}

The OU model assumes that a continuous trait evolves along a phylogeny according to an OU process. 
The process is equipped with a ``selection optimum" parameter $\mu$ which captures the optimal trait value; a ``selection strength" parameter $\alpha$ which represents the strength of the selection force that pulls the trait toward $\mu$; and the variance parameter $\sigma^2$ of the neutral drift.
The model has been used extensively to take into account natural selection in evolutionary studies \citep{beaulieu2012modeling, rohlfs2014modeling, uyeda2014novel, bastide_ho_baele_lemey_suchard_2021}.
Although the consistent property of estimators for $\mu$ and $\alpha$ have been studied thoroughly \citep{ho2013asymptotic, bartoszek2015phylogenetic, ane2017phase}, the consistency of ancestral state reconstruction methods under the OU model is not well-understood.
Here, we will fill in this gap. 
Applying Theorem \ref{thm:main}, we derive the following result:

\begin{theorem}
Under the OU model, the big bang condition is the necessary and sufficient condition for the existence of a consistent ancestral reconstruction method.
\end{theorem}

\begin{proof}
It is sufficient to check the regularity conditions \hyperlink{A1}{$(A_1)$} and \hyperlink{A2}{$(A_2)$} for the OU model.
Let $(Y_t)_{t \ge 0}$ be an OU process, we have
$$\mathbb{E}[Y(t) \mid Y(0)]= e^{-\alpha t}Y(0) + (1 - e^{-\alpha t})\mu,$$
and 
$$\text{Var}[Y(t) \mid Y(0)] = \frac{\sigma^2}{2 \alpha}(1 - e^{-2\alpha t}) \leq \sigma^2 t.$$
The condition \hyperlink{A1}{$(A_1)$} is satisfied with $C = \sigma^2; u(t) = e^{-\alpha t}; v(t) = (1 - e^{-\alpha t})\mu; \phi(y) = y$.
On the other hand, $P_{\rho_1, \tree}$ and $P_{\rho_2, \tree}$ are multivariate normal distributions.
Therefore, the condition \hyperlink{A2}{$(A_2)$} is trivial since the overlapped support is $\mbb{R}^{\lvert \partial \tree \rvert}$.
\end{proof}

\subsection{Reflected Brownian motion (RBM) model}

A limitation of both BM and OU models is that they cannot accommodate hard bounds on trait values.
Unfortunately, hard bounds do exist in nature.
For example, morphological measurements, such as body size and body mass, can only take positive values.
Some trait values are proportion (e.g., allele frequencies and genomic GC content) and thus are bounded between $0$ and $1$. 
\cite{boucher2016inferring} propose the Bounded Brownian motion (BBM) model, which assumes traits evolve under BM with two reflecting boundaries.
A particular case of this model for traits with positive values is the RBM model \citep{boucher2016inferring}.
Recall that if $X(t)$ is a BM starting from $X(0) > 0$, then $Y(t) = \lvert X(t) \rvert$ is a RBM starting from $Y(0) = X(0)$.
That is, the RBM model assumes that traits evolve according to a BM with a single reflecting boundary at zero.
We have the following theorem:

\begin{theorem}
Under the RBM model, the big bang condition is the necessary and sufficient condition for the existence of a consistent ancestral reconstruction method.
\end{theorem}

\begin{proof}
Again, we only need to verify that the RBM model satisfies the regularity conditions \hyperlink{A1}{$(A_1)$} and \hyperlink{A2}{$(A_2)$}.
Note that $X(t) \mid X(0)$ follows a normal distribution $\mc{N}(X(0), \sigma^2 t)$ and its moment generating function is $\psi(s) = \exp(X(0) s + \sigma^2 t s^2 / 2)$.
Therefore,
\begin{align*}
\mathbb{E}[X(t)^2 \mid X(0)] &= \frac{\partial^2 \psi}{(\partial s)^2}(0) = X(0)^2 + \sigma^2 t \\
\mathbb{E}[X(t)^4 \mid X(0)] &= \frac{\partial^4 \psi}{(\partial s)^4}(0) = X(0)^4 + 6 \sigma^2 t X(0)^2 + 3 \sigma^4 t^2.
\end{align*}
Since $Y^2(t) = X^2(t)$ and $Y(0) = X(0)$, we have
$$\mathbb{E}[Y(t)^2 \mid Y(0)] = Y(0)^2 + \sigma^2 t,$$
and 
\begin{align*}
\text{Var}[Y(t)^2 \mid Y(0)] &= \mbb{E}[Y(t)^4 \mid Y(0)] - (\mbb{E}[Y(t)^2 \mid Y(0)])^2 \\
&= 4 \sigma^2 t Y(0)^2 + 2 \sigma^4 t^2 \\
&\leq (4 \sigma^2 Y(0)^2 + 2 \sigma^4 h^*) t, \quad \forall t \in [0, h^*].
\end{align*}
Thus, the condition \hyperlink{A1}{$(A_1)$} is satisfied with $C = 4 \sigma^2 Y(0)^2 + 2 \sigma^4 h^*; u(t) = 1; v(t) = \sigma^2 t; \phi(y) = y^2$.
The condition \hyperlink{A2}{$(A_2)$} is trivial.
\end{proof}

%Verifying the condition \hyperlink{A1}{$(A_1)$} for the BBM model remains open.
%We hypothesize that under this model, the big bang condition is also the necessary and sufficient condition for the existence of a consistent ancestral reconstruction method.
%In Section \ref{sec:sim}, we use a simulation study to provide supporting evidence for this hypothesis.

\subsection{Bounded Brownian motion (BBM) model}

The BBM model \citep{boucher2016inferring} assumes that traits evolve under BM with two reflecting boundaries.
For simplicity, we assume that the BM is bounded in $[0,1]$.
We have the following theorem:

\begin{theorem}
Under the BBM model, the big bang condition is the necessary and sufficient condition for the existence of a consistent ancestral reconstruction method.
\label{thm:BBM}
\end{theorem}

\begin{proof}
We will verify the condition \hyperlink{A1}{$(A_1)$} with $C = \sigma^2 \pi^2; u(t) = \exp(- \sigma^2 t \pi^2/2); v(t) = 0; \phi(y) = \cos(\pi y)$.
First, we recall that the density function of $X(t) \mid X(0) = x_0$ is \citep{boucher2016inferring}
\[
p(x, x_0, t) = \frac{1}{\sqrt{2 \pi t} \sigma} \left \{  \sum_{k=-\infty}^{\infty}{\left [  \exp \left ( \frac{- (x - x_0 - 2k)^2}{2 \sigma^2 t} \right ) + \exp \left ( \frac{- (x +x_0 - 2k)^2}{2 \sigma^2 t} \right ) \right ]} \right \}.
\]
Therefore,
\[
\mathbb{E}[\cos(\pi X(t)) \mid X(0)] = \int_0^1{\cos(\pi x) p(x, X(0), t) dx}.
\]
Note that
\[
\int_0^1{\cos(\pi x) \exp \left ( \frac{- (x - x_0 - 2k)^2}{2 \sigma^2 t} \right ) dx} = \int_{-2k}^{-2k +1}{\cos(\pi x) \exp \left ( \frac{- (x - x_0)^2}{2 \sigma^2 t} \right ) dx}
\]
and
\begin{align*}
\int_0^1{\cos(\pi x) \exp \left ( \frac{- (x + x_0 - 2k)^2}{2 \sigma^2 t} \right ) dx} &= \int_{-2k}^{-2k +1}{\cos(\pi x) \exp \left ( \frac{- (x + x_0)^2}{2 \sigma^2 t} \right ) dx} \\
&= \int_{2k - 1}^{2k}{\cos(\pi x) \exp \left ( \frac{- (x - x_0)^2}{2 \sigma^2 t} \right ) dx}
\end{align*}
Hence
\begin{align*}
\mathbb{E}[\cos(\pi X(t)) \mid X(0)] &= \frac{1}{\sqrt{2 \pi t} \sigma} \int_{-\infty}^{\infty}{\cos(\pi x) \exp \left ( \frac{- (x - X(0))^2}{2 \sigma^2 t} \right ) dx} \\
&= \Re \left ( \frac{1}{\sqrt{2 \pi t} \sigma} \int_{-\infty}^{\infty}{e^{i \pi x} \exp \left ( \frac{- (x - X(0))^2}{2 \sigma^2 t} \right ) dx} \right) \\
&= \Re \left ( e^{i \pi X(0) - \sigma^2 t \pi^2/2} \right ) \\
&= \exp(- \sigma^2 t \pi^2/2) \cos(\pi X(0))
\end{align*}
where $\Re(\cdot)$ is the real part.
Similarly, we have
\[
\mathbb{E}[\cos^2(\pi X(t)) \mid X(0)] = \mathbb{E} \left [\frac{1 + \cos(2 \pi X(t))}{2} \mid X(0) \right ] = \frac{1 + \exp(- 2 \sigma^2 t \pi^2) \cos(2 \pi X(0))}{2}.
\]
Thus,
\begin{align*}
\text{Var}[\cos(\pi X(t)) \mid X(0)] &= \mathbb{E}[\cos^2(\pi X(t)) \mid X(0)] - (\mathbb{E}[\cos(\pi X(t)) \mid X(0)])^2 \\
&= \frac{1 + \exp(- 2 \sigma^2 t \pi^2) \cos(2 \pi X(0))}{2} - \exp(- \sigma^2 t \pi^2) \cos^2(\pi X(0)) \\
&= \frac{1}{2}\left ( 1 -  \exp(- \sigma^2 t \pi^2)\right ) \left [1 - \exp(- \sigma^2 t \pi^2) \cos(2 \pi X(0)) \right ] \\
& \leq \frac{1}{2} \left ( 1 -  \exp(- \sigma^2 t \pi^2)\right ) 2 \leq \sigma^2 t \pi^2.
\end{align*}
We conclude that the condition \hyperlink{A1}{$(A_1)$} is satisfied with $C = \sigma^2 \pi^2; u(t) = \exp(- \sigma^2 t \pi^2/2); v(t) = 0; \phi(y) = \cos(\pi y)$. The condition \hyperlink{A2}{$(A_2)$} is trivial.
\end{proof}

\subsection{Cox-Ingersoll-Ross (CIR) model}

The CIR model is an evolution model for traits that have positive values.
This model has been utilized for modelling evolutionary rate \citep{lepage2006continuous}, longevity of carnivores and ungulates \citep{blomberg2020beyond}, and rate of adaptive trait evolution \citep{jhwueng2020modeling}.
Similar to the OU process, the CIR process has a ``selection optimum" parameter $\mu$ and a ``selection strength" parameter $\alpha$.
Let $Y(t)$ be a CIR process, then $Y(t)$ follows the following stochastic differential equation:
\[
dY(t) = \alpha (\mu - Y(t)) dt + \sigma \sqrt{Y(t)} dB(t),
\]
where $B(t)$ is the standard BM.
We have
$$\mathbb{E}[Y(t) \mid Y(0)] = Y(0) e^{-\alpha t} + \mu (1 - e^{-\alpha t}),$$
and 
\begin{align*}
\text{Var}[Y(t) \mid Y(0)] &= Y(0) \frac{\sigma^2}{\alpha}(e^{-\alpha t} - e^{-2 \alpha t}) + \frac{\mu \sigma^2}{2 \alpha} (1 - e^{- \alpha t})^2 \\
& \leq [Y(0) + \mu/2] \frac{\sigma^2}{\alpha} (1 - e^{- \alpha t}) \\
&\leq [Y(0) + \mu/2] \sigma^2 t, \quad \forall t \in [0, h^*].
\end{align*}
Hence, the condition \hyperlink{A1}{$(A_1)$} holds with $C = [Y(0) + \mu/2] \sigma^2; u(t) = e^{-\alpha t}; v(t) = \mu (1 - e^{-\alpha t}); \phi(y) = y$.
Again, the condition \hyperlink{A2}{$(A_2)$} is trivial.
By Theorem \ref{thm:main}, we have:
\begin{theorem}
Under the CIR model, the big bang condition is the necessary and sufficient condition for the existence of a consistent ancestral reconstruction method.
\end{theorem}

\subsection{Simulation under the bounded Brownian motion (BBM) model}
\label{sec:sim}

In this section, we will use simulations to illustrate the result in Theorem \ref{thm:BBM}.
Specifically, we consider two sequences of bifurcating ultrametric caterpillar trees.
Recall that the distance from the root to the leaves of an ultrametric tree is constant and a caterpillar tree has a path, called the main path, that contains all internal nodes.
For both sequences, we create the $n$-th tree by adding a new leaf directly to the main path of the $(n-1)$-th tree such that the distance from the $n$-th leaf to the root is $1$.
The first sequence of trees satisfies the big bang condition: the $n$-th leaf attaches to the tree at the $n$-th internal node, whose distance to the root is $1/n$ (see Figure \ref{fig:caterpillartrees} -- left).
On the other hand, the second sequence does not satisfy the big bang condition: the $n$-th leaf attaches to the tree at the $n$-th internal node, whose distance to the root is $1 - 1/n$ (see Figure \ref{fig:caterpillartrees} -- right).

\begin{figure}[h]
    \centering
    \includegraphics[width=0.49\textwidth]{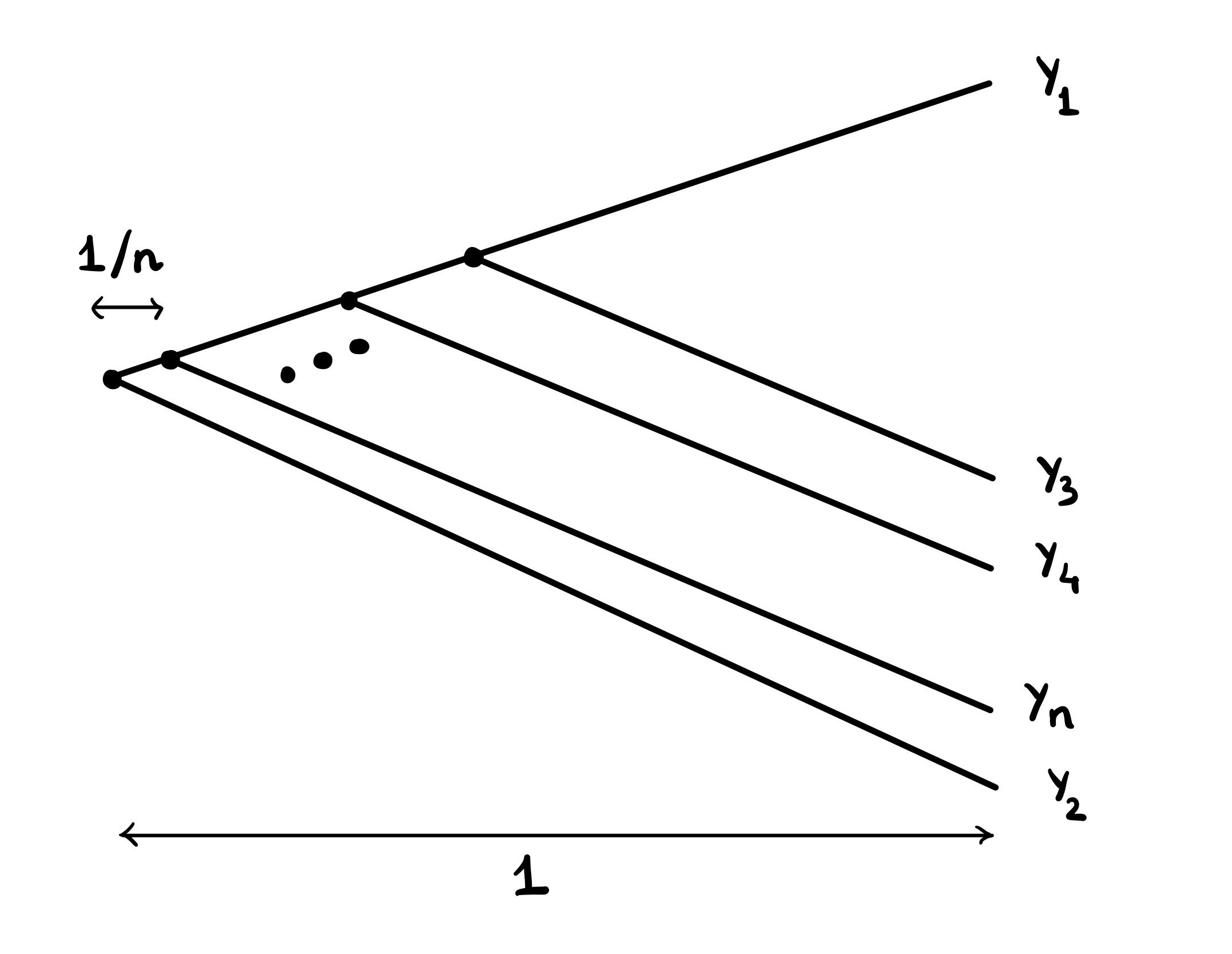}
    \includegraphics[width=0.49\textwidth]{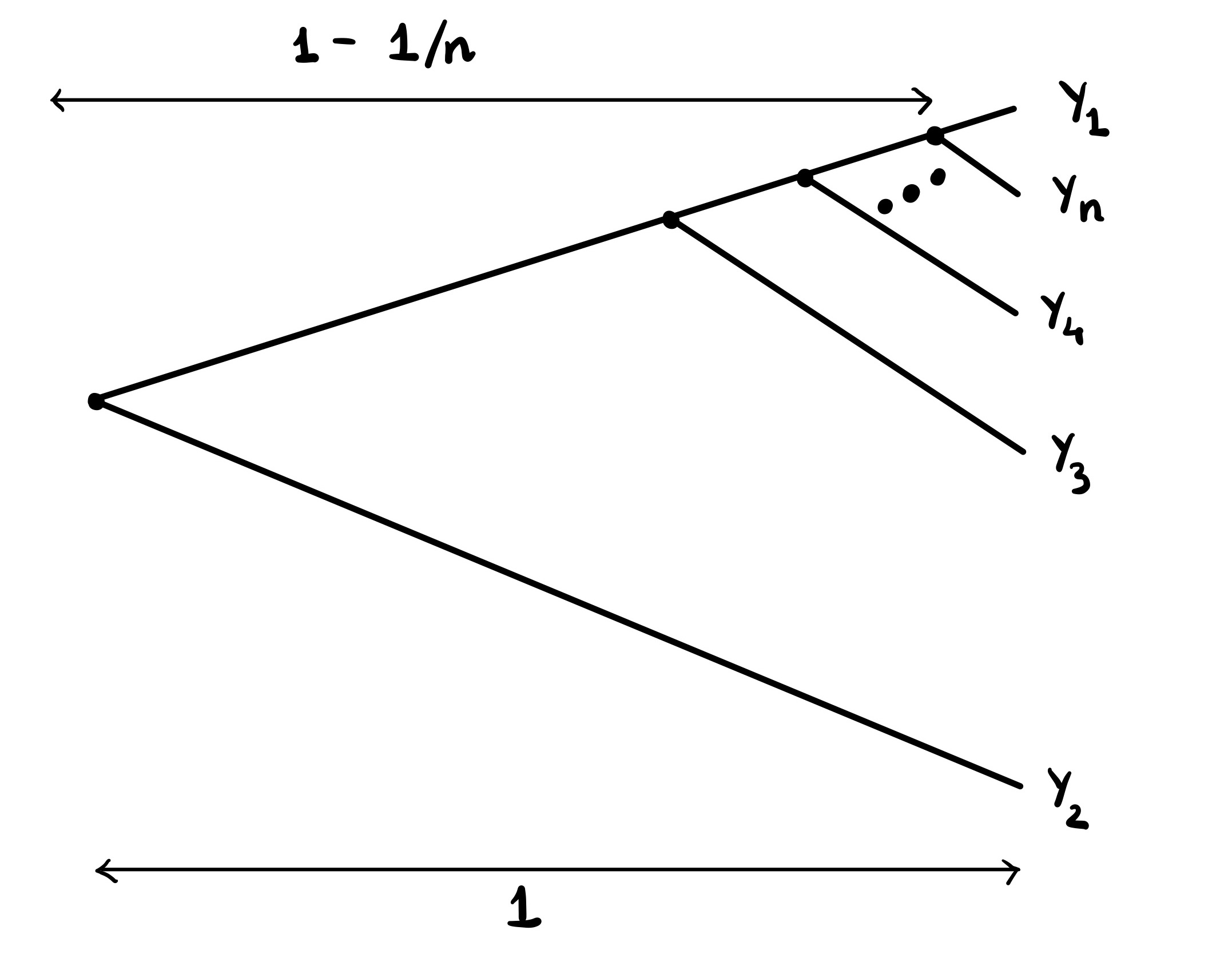}
    \caption{Two sequences of caterpillar trees consider in the simulation. Left: the big bang condition is satisfied (new branch is closer and closer to the root). Right: the big bang condition is not satisfied (all branches are far away from the root). }
    \label{fig:caterpillartrees}
\end{figure}

In this simulation, we use the \texttt{R} functions provided in \citet{boucher2016inferring} for simulating and fitting under the BBM model \footnote{https://github.com/fcboucher/BBM}.
We focus on trees with size $n = 10, 100, 1000$. 
So, we have $6$ trees in total ($3$ trees for each sequence).
For each tree, we simulate the trait values at the leaves $1000$ times under the BBM model with the ancestral state $\rho = 0$, variance $\sigma^2 = 1/2$, and two bounds $\pm 1$ using the \texttt{R} function \texttt{Sim\_BBM}.
Then, we use the function \texttt{fit\_BBM\_model\_uncertainty} to fit the BBM model and return the MLE of the ancestral state $\rho$.
When the big bang condition holds, the estimate $\hat \rho$ is more precise as the number of species increases (see Figure \ref{fig:boxplots} -- left).
On the other hand, when the big bang condition is not satisfied, the precision of $\hat{\rho}$ stays the same for all trees.
%This result provides strong supporting evidence to our hypothesis: the big bang condition is required for consistently estimating the ancestral state under the BBM model.

\begin{figure}[h]
    \centering
    \includegraphics[width=0.45\textwidth]{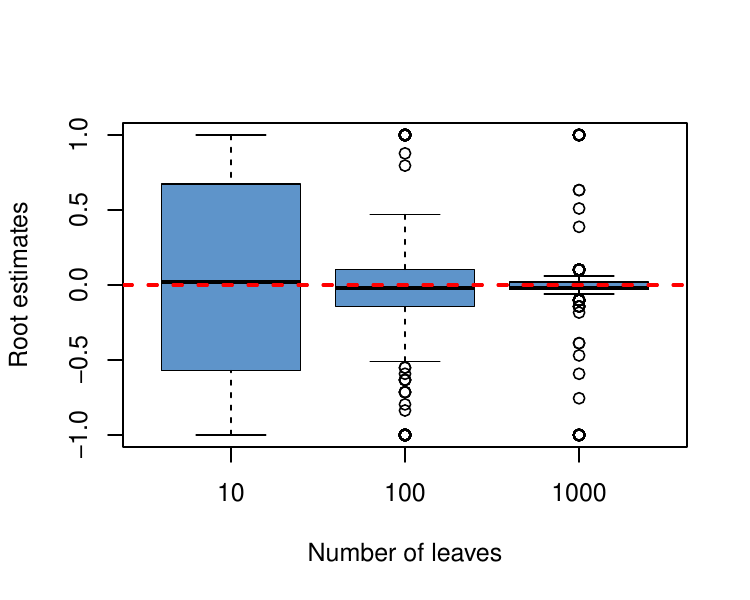}
    \includegraphics[width=0.45\textwidth]{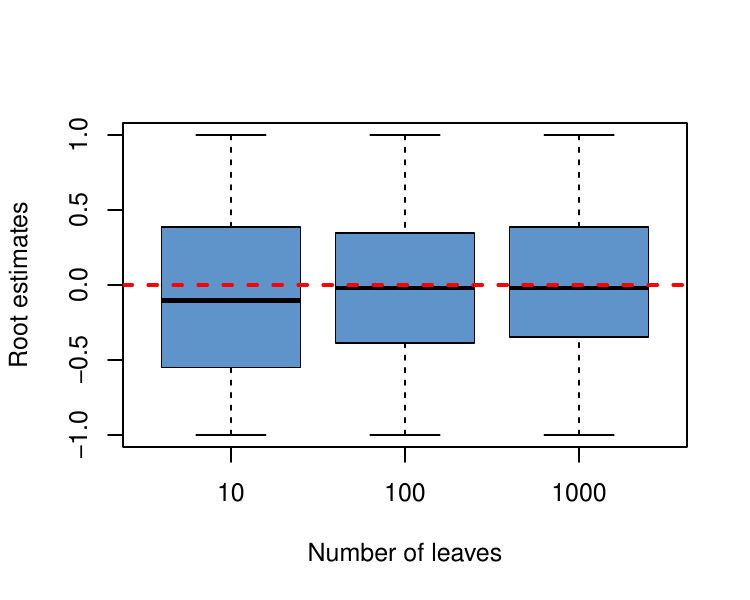}
    \caption{Estimates of the ancestral state $\rho$ under the BBM with $\rho = 0$, $\sigma^2 = 1/2$, and two bounds $\pm 1$. Left (the big bang condition holds):  the estimates become more accurate as the number of leaves increases. Right (the big bang condition does not hold): the accuracy is the same for all trees.}
    \label{fig:boxplots}
\end{figure}

\section{Discussion and Conclusion}

In this paper, we prove that under some regularity conditions, the big bang condition is a necessary and sufficient condition for the existence of a consistent ancestral state reconstruction method for continuous traits.
We verify these conditions for Ornstein-Uhlenbeck, reflected Brownian motion, bounded Brownian motion and Cox-Ingersoll-Ross models.

It is worth noticing that under the BM model, the MLE for the ancestral state is consistent if and only if the big bang condition holds \citep{ho2021can}.
However, it is unclear if this result is still true beyond BM models.
Since MLE is the most popular method for reconstructing the ancestral state, studying its consistency property is of great interest.

The results in this paper assume that both the tree topology and branch lengths are known.
However, we may know the tree topology but not branch lengths in practice.
Many popular tree reconstruction methods such as maximum parsimony, neighbour-joining, and quartet puzzling only return the tree topology.
A recent study shows that for discrete traits, the big bang condition may not guarantee the existence of a consistent ancestral state reconstruction method when branch lengths are unknown \citep{ho2022ancestral}.
An open question is whether the big bang condition is still a sufficient condition for the existence of a consistent estimator for the ancestral state of continuous traits.

\section*{Data Availability Statements}

The R code for the simulations is available at \url{https://github.com/lamho86/When-can-we-reconstruct-the-ancestral-state-Beyond-Brownian-motion}.

\section*{Conflict of interest}
The authors have no competing interests to declare.

\section*{Acknowledgement}

LSTH was supported by the Canada Research Chairs program, the NSERC Discovery Grant RGPIN-2018-05447, and the NSERC Discovery Launch Supplement DGECR-2018-00181. VD was supported by a startup fund from the University of Delaware and National Science Foundation grant DMS-1951474.

\bibliographystyle{chicago}

\bibliography{references}
\end{document}